\documentclass[10pt]{article}
\usepackage[latin9]{inputenc}
\usepackage[T1]{fontenc}
\usepackage[french,english]{babel}
\usepackage{graphicx}
\usepackage[top=3cm, bottom=3cm, left=3cm, right=3cm]{geometry}
\usepackage{float}
\usepackage{amssymb}
\usepackage{amsmath}
\usepackage{titlepic}
\usepackage{caption}
\usepackage{epstopdf}
\usepackage{xfrac}
\usepackage{enumitem}
\usepackage{subfigure}
\usepackage{bm}
\usepackage{esint}
\usepackage{mathtools}
\usepackage[dvipsnames]{xcolor}

\usepackage{amsthm}

 \def\dr{\partial}
\newcommand{\nm}{\noalign{\smallskip}}

\newcommand{\R}{\mathbb{R}}
\newcommand{\N}{\mathbb{N}}

\newcommand{\C}{\mathbb{C}}

\newcommand{\nubf}{\boldsymbol{\nu}}

\newcommand{\hf}{\widehat{f}}

\newcommand{\Bbf}{\mathbf{B}}
\newcommand{\Abf}{\mathbf{A}}
\newcommand{\fbf}{\mathbf{f}}
\newcommand{\pbf}{\mathbf{p}}

\newcommand{\vbf}{\mathbf{v}}
\newcommand{\ebf}{\mathbf{e}}

\newcommand{\Dbf}{\mathbf{D}}

\newcommand{\Ibf}{\mathbf{I}}

\newcommand{\Ebf}{\mathbf{E}}
\newcommand{\Gabf}{\mathbf{\Gamma}}
\newcommand{\Hbf}{\mathbf{H}}

\newcommand{\dd}{\mathrm{d}}
\newcommand{\tin}{\text{ in }}
\newcommand{\ton}{\text{ on }}

\newtheorem{theorem}{Theorem}[section]
\newtheorem{remark}{Remark}[section]
\newtheorem{corollary}{Corollary}[section]
\newtheorem{definition}{Definition}[section]
\newtheorem{lemma}{Lemma}[section]

\newtheorem{proposition}{Proposition}[section]

\newtheorem{cond}{Condition}

\begin{document}

\title{Modal approximation for strictly convex plasmonic resonators in the time domain: the Maxwell's equations \footnote{This work was supported in part by the Swiss National Science Foundation grant number
200021--172483.}}
\author{Habib Ammari\thanks{\footnotesize Department of Mathematics,
ETH Z\"urich,
R\"amistrasse 101, CH-8092 Z\"urich, Switzerland (habib.ammari@math.ethz.ch; alice.vanel@sam.math.ethz.ch).}   \and Pierre Millien\thanks{\footnotesize ESPCI Paris, PSL University, CNRS, Sorbonne Universit\'e, Institut Langevin 1 rue Jussieu, 75005 Paris, France (pierre.millien@espci.fr). }\and Alice L. Vanel\footnotemark[2] }

%
%
%

\date{}
\maketitle

\begin{abstract}
We study the possible expansion of the electromagnetic field scattered by a strictly convex metallic nanoparticle with dispersive material parameters placed in a homogeneous medium in a low-frequency regime as a sum of \emph{modes} oscillating at complex frequencies (diverging at infinity), known in the physics literature as the \emph{quasi-normal modes} expansion. We show that such an expansion is valid in the static regime and that we can approximate the electric field with a finite number of modes. We then use perturbative spectral theory to show the existence, in a certain regime, of plasmonic resonances as poles of the resolvent for Maxwell's equations with non-zero frequency. We show that, in the time domain, the electric field can be written as a sum of modes oscillating at complex frequencies. We introduce renormalised quantities that do not diverge exponentially at infinity. 
\end{abstract}

\def\keywords2{\vspace{.5em}{\textbf{  Mathematics Subject Classification
(MSC2000).}~\,\relax}}
\def\endkeywords2{\par}
\keywords2{35Q61, 35C20, 31B10}

\def\keywords{\vspace{.5em}{\textbf{ Keywords.}~\,\relax}}
\def\endkeywords{\par}
\keywords{plasmonic resonance, normal modes, dispersive scatterer, time domain expansion}

\section{Introduction}

\subsection{Position of the problem}
Modal analysis has been a useful tool in wave physics to understand the behaviour of complex systems and to numerically compute the response to an excitation.
For a bounded, lossless system, the operator $\Delta^{-1}$ associated with the wave equation with Dirichlet or Neumann boundary conditions is compact and self-adjoint when studied in the right functional spaces. Hence it can be diagonalised and a complete basis of eigenmodes with real eigenfrequencies can be exhibited.
The response of the system to an excitation can then be computed by summing the response of each mode to the excitation.

However, when the system exhibits loss (by absorption or radiation), the operator cannot be diagonalised with the classical spectral theorem and the eigenfrequencies have a negative imaginary part. Using sophisticated micro-local analysis and building on the Lax-Phillips scattering theory \cite{lax1964scattering}, several authors have obtained \emph{resonance expansions} in various cases.

We refer to~\cite{ramm1982mathematical,zworski1999resonances} for a general presentation of resonance expansions and to the recent book \cite{dyatlov2019mathematical} for the state-of-the-art. These expansions rely on \emph{high-frequency estimates} of the resolvent and, to the best of our knowledge, rigorous resonance expansions for the transmission problem in unbounded domains have so far only been obtained for the local scalar wave equation~\cite{popov1999resonances,popov1999distribution} with non-negative coefficients until the recent progress of~\cite{cassier2017spectral,cassier2017mathematical}, which deals with \emph{non-locality} in time. One should also note that the spectral analysis used in these papers is far from elementary and hardly accessible to the non-specialist of the subject. 

Nevertheless, in the physics community modal analysis is heavily used in numerical nano-photonics (see the review paper~\cite{lalanneReview} and references therein) to study the interaction of light with resonant structures such as \emph{nanoparticles} and \emph{metamaterials}, which are described by the \emph{non-local} Maxwell's equations~\cite{binkowski19}. In practice, modes (or generalised eigenvectors) are computed by solving, in the frequency domain, the source-free Maxwell's equations satisfying the outgoing radiation condition~\cite{Zolla:18}.  Several practical and theoretical issues naturally arise with this approach. Fields oscillating at a complex frequency with negative imaginary part solving the radiation condition diverge exponentially in space (Lamb's so-called \emph{exponential catastrophe}~\cite{sirenko2007modeling}), causing physical interpretations to be complicated and numerical computations in large domains to be problematic without renormalisation techniques~\cite{stout2019eigenstate}. Generalised modes are not an orthogonal family, rendering energy considerations difficult. The density of the linear span of the family of modes has not been shown, raising questions about the possibility to represent any electromagnetic field as a sum of modes.

In this article, we study the electromagnetic field scattered by a strictly convex metallic nanoparticle in a \emph{low-frequency regime}, which is the one relevant for applications as it roughly corresponds to the visible/infra-red frequency range.
We show how the study of the \emph{Electric Field Integral Equation} inside the particle can lead to an approximated pole expansion of the resolvent for a particular \emph{low frequency} perturbative regime where ratio of the size of the particle over the wavelength of the excitation field is small but not zero. 
The analysis is based on the spectral theory of the  \emph{Neumann-Poincar\'e operator} and the classic perturbation theory from Kato.

\subsection{Main contributions}

This paper is a follow up of \cite{baldassari2021modal} where the case of scalar waves was studied.
The starting point is the \emph{Electric Field Integral Equation} (or Lippman-Schwinger equation)
\begin{align*}
\left(I -  \gamma^{-1}(\omega)\mathcal{T}^{\omega}\right)\Ebf =\Ebf^\text{in} \qquad  \tin D,
\end{align*}
where $\mathcal{T}^\omega$ is a singular integral operator and $\gamma$ is a non-linear function of $\omega$ that depends on the permittivity model for the scatterer $D$.
Building on the previous spectral analysis  of $\mathcal{T}$ \cite{costabel2012essential}  and classic  results on the compact symmetrisable \emph{Neumann-Poincar\'e operator}, we exhibit a complete modal basis for the static ($\omega=0$) transmission problem (Theorem~\ref{theo:staticmodal}). We show that, under the assumption that the particle is strictly convex, the excitation coefficients of the eigenfunctions exhibit superpolynomial decay with the order of the mode, and therefore the field can be well approximated by a finite number of modes.
Using elementary perturbative analysis, we then show in Proposition~\ref{prop:modalinside} that the resolvent for the dynamic problem ($\omega\not=0$) can be approximated by a perturbed resolvent of a finite-dimensional operator (the truncated static operator).
Next, using Rouch\'e's theorem, we prove the existence of poles for this approximated resolvent and give a \emph{resonance-like} expansion for the electric field (Theorem~\ref{theo:mainexpansion}) inside the particle. From this expansion we construct the so-called \emph{quasi-normal modes} found in the physics literature.
Finally, using elementary complex analysis tools, we give an expansion for the \emph{low-frequency part} of the electromagnetic field in the time domain in Theorem~\ref{theo:resonanceexpansion}. By doing so, we show that in the time domain causality ensures that the electromagnetic field does not diverge exponentially in space. Similar results were obtained in~\cite{colom2018modal} for a non-dispersive dielectric spherical scatterer in any frequency range.

\section{Maxwell's equations for a metallic resonator}
\subsection{Problem geometry}
We are interested in the scattering problem of an incident spherical wave on a plasmonic nanoparticle. The homogeneous medium is characterised by the electric permittivity $\varepsilon_m$ and the magnetic permeability $\mu_m$. Let $D$ be a smooth bounded domain in $\mathbb{R}^3$, of class $C^\infty$, characterised by electric permittivity $\varepsilon_c$ and magnetic permeability $\mu_c$. We assume the particle to be non-magnetic, i.e., $\mu=\mu_c = \mu_m$ in $\mathbb{R}^3$. We define the wavenumbers $k_c=\omega\sqrt{\varepsilon_c\mu_c}$ and $k_m=\omega\sqrt{\varepsilon_m\mu_m}$. Let $\varepsilon=\varepsilon_c \chi(D)+\varepsilon_m \chi(\mathbb{R}^3 \setminus \bar{D})$, where $\chi$ denotes the characteristic function. We denote by $c_0$ the speed of light in vacuum, $c_0=1/\sqrt{\varepsilon_0\mu_0}$, and by $c$ the speed of light in the medium, $c=1/\sqrt{\varepsilon_m\mu_m}$. Let $D=z+\delta B$, where $B$ is the reference domain and contains the origin and $D$ is located at $z\in\mathbb{R}^3$ and has a characteristic size $\delta$ small compared to the operating wavelength $\delta k_m\ll 1$. Let $\nu$ be the normal vector. Throughout this paper, we assume that $\varepsilon_m$ and $\mu_m$ are real and positive. We also assume that $\Im{\varepsilon_c} \leq 0$.

Hereafter we use the Drude model~\cite{ordal83} to express the electric permittivity of the particle:
\begin{equation}\label{eq:drude}
\varepsilon_c(\omega)=\varepsilon_0\left(1-\frac{\omega_p^2}{\omega^2+i\omega\mathrm{T}^{-1}}\right),
\end{equation}
where the positive constants $\omega_p$ and $\mathrm{T}$ are the plasma frequency and the collision frequency or damping factor, respectively. We write $\varepsilon_m=\sqrt{n}\varepsilon_0$ where $n$ is the refractive index of the medium. The use of this model is absolutely non-restrictive and one could use other models for the permittivity of metals \cite{johnson1972optical}. The Drude model makes for easier analytical expressions in section \ref{sec:statanddyn} but the results would still be valid with another physical model.

The results of subsection \ref{sec:sumtruncation}, which are central to this work and are used in all the subsequent sections are valid under the following condition:
\begin{cond}
The domain $D\subset \R^3$ has to be strictly convex.
\end{cond}

\subsection{Formulation}
For a given incident wave $(\Ebf^{\text{in}}, \Hbf^{\text{in}})$ solution to Maxwell's equations
 \begin{equation*}
 \left \{
 \begin{array}{ll}
\nabla\times{\Ebf^{\text{in}}} = i\omega\mu_m{\Hbf^{\text{in}}} \quad &\mbox{in } \mathbb{R}^3,\\
\nm
\nabla\times{\Hbf^{\text{in}}} = -i\omega\varepsilon_m {\Ebf^{\text{in}}} -i\dfrac{1}{\omega \mu_m} \mathbf{p} \delta_s\quad &\mbox{in
} \mathbb{R}^3,
 \end{array}
 \right .
 \end{equation*}
where the source at $s$ has a dipole moment $\mathbf{p}\in\R ^3$, let $(\Ebf,\Hbf)$ be the solution to the following Maxwell's equations:
\begin{equation}
\label{eq:maxwell} \left\{
\begin{array}{ll}
 \nabla \times \Ebf = i \omega \mu \Hbf &  \mbox{in} \quad \mathbb{R}^3\setminus \dr D, \\
 \nabla  \times \Hbf= - i \omega \varepsilon \Ebf & \mbox{in} \quad \R^3\setminus \dr D, \\
 {[}{\nubf} \times \Ebf]= [{ \nubf} \times \Hbf] = 0 & \mbox{on} \quad \dr D,
\end{array}
\right.
\end{equation}  subject to the Silver-M\"{u}ller radiation condition:
 $$\lim_{|x|\rightarrow\infty} |x| (\sqrt{\mu_m} (\Hbf- \Hbf^{\text{in}}) \times\hat{x}-\sqrt{\varepsilon_m} (\Ebf-\Ebf^{\text{in}}))=0,$$
 where $\hat{x} = x/|x|$. Here, $[{\nubf} \times{\Ebf}]$ and $[{\nubf} \times{\Hbf}]$ denote the jump of ${\nubf} \times{\Ebf}$
and ${\nubf} \times{\Hbf}$ along $\dr D$, namely,
 \begin{equation*}
 [\nubf\times{\Ebf}]=(\nubf\times \Ebf)\bigr|_+ -(\nubf \times \Ebf)\bigr|_-,\quad [\nubf\times{\Hbf}]=(\nubf\times \Hbf)\bigr|_+ -(\nubf\times \Hbf)\bigr|_-.
  \end{equation*}
We introduce the Sobolev spaces
\begin{align*}
H(\mathrm{curl}, D)&=\{\vbf \in L^2(D,\R^3),\, \mathrm{curl}\, \vbf \in L^2(D,\R^3)\}, \\
H_{\mathrm{loc}}(\mathrm{curl}, \R^3\setminus \overline{D})&=\{\vbf \in L_{\mathrm{loc}}^2(\R^3\setminus \overline{D}),\, \mathrm{curl}\, \vbf \in L_{\mathrm{loc}}^2(\R^3\setminus \overline{D})\}.
\end{align*}

\begin{proposition} If $\Im\left[\varepsilon_c\right]\neq 0$, then problem (\ref{eq:maxwell}) is well-posed.
Moreover, if we denote by $(\Ebf,\Hbf)$ its unique solution, then $(\Ebf,\Hbf)\big\vert_D \in H(\mathrm{curl}, D)$ and $(\Ebf,\Hbf)\big\vert_{\R^3\setminus D} \in H_{\mathrm{loc}}(\mathrm{curl},\R^3\setminus \overline{D})$.
\end{proposition}

\begin{proof}
The well-posedness is addressed in~\cite{torres1998maxwell,costabel2012essential,ammari2017mathematicalscalar}.
\end{proof}

\subsection{Volume integral equation for the electric field}\label{sec:VIE}
We now recall the well-known Lippmann-Schwinger equation~\cite{ammari03} satisfied by the electric field for a non-magnetic particle:
\begin{equation} \label{eq:lippmann}
\Ebf(x)= \Ebf^{\text{in}}(x)+\frac{\varepsilon_m-\varepsilon_c}{\varepsilon_m}\left(\frac{\omega}{c}\right)^2\int_D \Gabf^{\frac{\omega}{c}}(x,y) \Ebf(y)\dd y, \qquad x \in \mathbb{R}^3,
\end{equation}
where $\Gabf^\frac{\omega}{c}$, the dyadic Green's function, is defined in Appendix \ref{app:dyadic}.
Consequently, it suffices to derive an approximation for the electric field $\Ebf$ inside $D$ and insert it in the right-hand side of~\eqref{eq:lippmann} to obtain an expression for $\Ebf$ for all points outside.
Using the dyadic Green's function, one can express the incident field as
\begin{equation}
\Ebf^{\text{in}}(x)= \Gabf^\frac{\omega}{c} (x,s) \pbf, \qquad x \in \mathbb{R}^3.
\end{equation}
\begin{definition}
We denote the contrast $\gamma$ by
\begin{equation*} \label{eq:def_l}
\gamma(\omega)= \frac{\varepsilon_m}{\varepsilon_m-\varepsilon_c(\omega)}.
\end{equation*}
\end{definition}
Restricting equation~\eqref{eq:lippmann} to $D$ yields the following integral representation.
\begin{proposition} \label{prop:LSequation}
The electric field inside the particle satisfies the volume integral equation:
\begin{align}\label{eq:fieldinside}
\left(\gamma(\omega) I-  \mathcal{T}^{\frac{\omega}{c}}_D\right) \Ebf = \gamma(\omega) \Ebf^{\text{\emph{in}}} \qquad \tin D,
\end{align} where $\mathcal{T}_D^\frac{\omega}{c} : L^2(D,\R^3)\rightarrow L^2(D,\R^3)$ is a singular integral operator of the Calder\'on-Zygmund type, defined in Appendix \ref{de:defTk}.
\end{proposition}

\begin{proof} See~\cite[Chapter 9]{colton2012inverse} or~\cite{costabel2010volume}.
\end{proof}

Note that the operator $\mathcal{T}_D^{\frac{\omega}{c}}$ is neither compact nor self-adjoint on $L^2(D,\R^3)$ so diagonalising $\mathcal{T}_D^{\frac{\omega}{c}}$ directly to find a modal expansion for $\Ebf$ is not possible.
However, when $\omega=0$ (in the \emph{static regime}), the operator $\mathcal{T}_D^0$ has nice spectral properties.

\section{The electrostatic approximation}\label{sec:static}

\subsection{Rescaling of the problem}\label{sec:rescale}
To study the behaviour of the electromagnetic field when the particle becomes small with respect to the wavelength, it is convenient to write the integral equation directly on the reference, unit-size domain $B=(D-z)/\delta$, so that the size of the domain is fixed and only one asymptotic parameter remains, $\omega\delta c^{-1}$, which is the ratio of the particle size over the incoming electromagnetic field wavelength.

We write, for $x\in D$,  $x=\phi(\tilde{x})$ with $\tilde{x}\in B$ and $\phi:\, \tilde{x}\mapsto \delta\tilde{x}+z$  and perform the change of variable in the singular integral equation~\eqref{eq:fieldinside} (see \cite[p. 41]{mikhlin2014multidimensional} or \cite{seeley1959singular}). For any function $f$ in $D$ we denote by $\widetilde{f}=f\circ \phi$ the corresponding function in $B$. Equation~\eqref{eq:fieldinside} becomes:
\begin{align}\label{eq:rescaled}
\left(\gamma(\omega) I -  \mathcal{T}^{\frac{\omega\delta} {c}}_B\right) \widetilde{\Ebf}   = \gamma(\omega) \widetilde{\Ebf}^\text{in}  \qquad \tin B.
\end{align}

The goal of this section is to study the relationship between the solution $\widetilde{\Ebf}$ of~\eqref{eq:rescaled} and the solution $\widetilde{\Ebf}^0$ of the electrostatic problem
\begin{align}\label{eq:fieldinsidestatic}
\left(\gamma(\omega) I -  \mathcal{T}^{0}_B\right) \widetilde{\Ebf}^0   = \gamma(\omega) \widetilde{\Ebf}^\text{in} \qquad \tin B.
\end{align}

\subsection{Modal decomposition in the static regime}
The goal of this section is to establish Theorem~\ref{theo:staticmodal}. We start by recalling the orthogonal decomposition:
\begin{lemma}\label{lem:orthogdecomp}
We have
\begin{equation*}
L^2(B,\R^3) = \nabla H^1_0(B) \oplus \mathbf{H}_0(\text{\emph{div} } 0,B) \oplus \mathbf{W}(B),
\end{equation*}
where $\mathbf{H}_0(\text{\emph{div} } 0,B) $ is the space of divergence-free $L^2$ vector fields with vanishing boundary trace, and $\mathbf{W}(B)$ is the space of gradients of harmonic $H^1$-functions. We denote by $P_{\nabla H^1_0}$, $P_{\mathbf{H}_0}$ and $P_\mathbf{W}$ the orthogonal projections for the usual $L^2(B,\R^3)$ scalar product on $\nabla H^1_0(B)$, $\mathbf{H}_0(\text{\emph{div} } 0,B)$ and $\mathbf{W}(B)$ respectively.
\end{lemma}
We start with the following result from~\cite{costabel2012essential}:
\begin{proposition}\label{prop:decompT0}
$\mathcal{T}_B^0$ is a bounded self-adjoint map on $L^2(B,\R^3)$ with $  \nabla H^1_0(B)$, $\mathbf{H}_0(\text{\emph{div} } 0,B)$ and $\mathbf{W}(B)$ as invariant subspaces. On $\nabla H^1_0(\Omega)$, $\mathcal{T}_B^0[\ebf]=\ebf$, on $\mathbf{H}_0(\text{\emph{div} } 0,B)$, $\mathcal{T}_B^0[\ebf]=0$ and on $\mathbf{W}(B)$: $$\nubf\cdot \mathcal{T}_B^0[\ebf]= \left(\frac{1}{2}I - \mathcal{K}_B^*\right)[\ebf\cdot \nubf] \quad \ton \partial B , $$ where $\mathcal{K}_B^*$ is the Neumann-Poincar\'e operator defined in Section \ref{sec:app_def}.
\end{proposition}
\begin{proof}
The proof can be found in~\cite{friedman1984spectral, costabel2012essential}. 
\end{proof}
\begin{corollary}\label{cor:linkk*}
\noindent Let $\gamma \not=1$. Let $\ebf\not\equiv 0$ be such that
\begin{align*}
\gamma \ebf - \mathcal{T}_B^0[\ebf]= 0 \quad \tin B.
\end{align*}
Then,
\begin{align*}
&\ebf\in \mathbf{W}(B),\\
&\nabla \cdot \ebf  = 0 &\tin B, \\
&\gamma \ebf = \nabla \mathcal{S}_B[\ebf \cdot \nubf]&\tin B , \\
&\gamma \ebf \cdot \nubf = \left(\frac{1}{2}I - \mathcal{K}_B^*\right)[\ebf\cdot \nubf] &\ton \partial B,
\end{align*}
where $\mathcal{S}_B$ is the single-layer potential defined in Section \ref{sec:app_def}.
\end{corollary}
\begin{remark}It has been shown in~\cite{ammari2016surface,ammari2017mathematicalscalar} that the plasmonic resonances are linked to the eigenvalues of the Neumann-Poincar\'e operator. Corollary~\ref{cor:linkk*} shows that the volume integral approach and the surface integral approach are consistent with one another.
\end{remark}
\begin{remark}Note that our definition of the Green's function for the Helmholtz equation has an additional minus sign compared to the one given in \cite{costabel2012essential}, hence there are some disparities in the formulae. 
\end{remark}
We now recall a few classical results on the Neumann-Poincar\'e operator:
\begin{proposition}If $\partial B$ has $\mathcal{C}^{1,\alpha}$ regularity for $\alpha>0$ then $\mathcal{K}_B^*$ is a compact operator.
\end{proposition}
\begin{proof}See~\cite[Chapter 2]{ammari2013mathematical}.
\end{proof}
\begin{proposition}[Plemelj symmetrisation principle] \label{prop:symmetrisation}Let $\mathcal{H}^*(\partial B)$ be the Hilbert space $H^{-1/2}(\partial B)$ equipped with the following inner product:
\begin{align*}
\left\langle u, v\right\rangle_{\mathcal{H}^*}= -\left\langle u, \mathcal{S}_B[v] \right\rangle_{-1/2,1/2}.
\end{align*}
Then $\mathcal{K}_B^*$ is a self-adjoint operator on $\mathcal{H}^*(\partial B)$.
\end{proposition}
\begin{proof} This is a direct consequence of Lemmas~\ref{lem:calderon} and~\ref{lem:unitary}.
\end{proof}
\begin{theorem}[Diagonalisation of $\mathcal{K}_B^*$]
$\mathcal{K}_B^*$ has a discrete set of real eigenvalues $(\lambda_n)_{n\in \N}$ with associated eigenvectors $(\phi_n)_{n\in \N}$ and
\begin{align*}
\mathcal{K}_B^*=\sum_{n=0}^\infty \lambda_n \left\langle \phi_n,\cdot \right\rangle_{\mathcal{H}^*(\partial B)} \phi_n,
\end{align*}
with $\lambda_n \in~]-1/2,1/2]$, $\lambda_0=1/2$ and $|\lambda_n| \rightarrow 0$ as $n \rightarrow +\infty$.
\end{theorem}
\begin{proof}See~\cite[Chapter 2]{ammari2013mathematical}.
\end{proof}

We can now establish the spectral decomposition for $\mathcal{T}_B^0$ on $\mathbf{W}(B)$:
\begin{proposition}[Spectral decomposition of $\mathcal{T}_B^0$]\label{prop:eigenbasis} The set of eigenvalues $(\gamma_n)_{n\in \N}$ of $\mathcal{T}^0_B\big\vert_\mathbf{W}$ is discrete, and the associated eigenvectors $(\ebf_n)_{n\in \N}$ form an orthonormal basis of $\mathbf{W}(B)$. Hence we have:
\begin{align*}
\mathcal{T}_B^0\big\vert_\mathbf{W} = \sum_{n=0}^\infty \gamma_n \langle \ebf_n, \cdot \rangle_{L^2(B,\R^3)} \ebf_n,
\end{align*}
and $\gamma_n\in~]0,1].$
\end{proposition}
\begin{proof}
The eigenvector basis of $\mathbf{W}(B)$ for $\mathcal{T}_B^0$ can be constructed from the basis $(\phi_n)_{n\in \N}$ of $H^{-1/2}(\partial B)$  by setting $\ebf_n=(\gamma_n)^{-1}\nabla \mathcal{S}_B[\phi_n]$ and $\gamma_n=1/2-\lambda_n$, as seen in Corollary~\ref{cor:linkk*}.
\end{proof}

 \begin{proposition}\label{lem:decompR0} If we denote by $\mathcal{R}^0_B$ the resolvent of $\mathcal{T}^0_B$ on $L^2(B,\R^3)$ then \begin{align*}
\mathcal{R}^0_B(\gamma)=\frac{1}{\gamma -1} P_{\nabla H^1_0}  + \frac{1}{\gamma} P_{\mathbf{H}_0} + \sum_{n=0}^\infty \frac{\left\langle \cdot , \ebf_n \right\rangle_{L^2(B,\R^3)}}{\gamma-\gamma_n} \ebf_n.
\end{align*}
The essential spectrum of $\mathcal{T}^0_B$ is $\sigma_{ess}=\{0,\frac{1}{2},1\}$. $0$ and $1$ are  isolated eigenvalues of infinite multiplicity, while $\frac{1}{2}$ is an accumulation point.
\end{proposition}
\begin{proof} This is a direct consequence of Propositions~\ref{prop:decompT0} and~\ref{prop:eigenbasis}.
\end{proof}

\begin{remark}[Scale invariance of $\mathcal{T}_B^0$]\label{rem:scaleinvariance}
The operator $\mathcal{T}_B^0$ is invariant under scaling. If one sets $x=\delta \tilde{x}+z$ with $\tilde{x}\in B$ and $\phi(x) = \widetilde{\phi}(\tilde{x})$ for any $\widetilde{\phi} \in L^2(B,\R^3)$ then
\begin{align*}
\mathcal{T}^0_{\delta B}[\phi](x) = \mathcal{T}^0_B[\widetilde{\phi}](\widetilde{x}).
\end{align*}
To keep the notations simple, we will denote by $\ebf_n$ the eigenvectors of $\mathcal{T}^0_B$ in $\mathbf{W}(B)$ as well as the eigenvectors of $\mathcal{T}^0_D$ in $\mathbf{W}(D)$.
\end{remark}

\begin{theorem}[Modal decomposition in the static regime]\label{theo:staticmodal}

\begin{align}\label{eq:staticmodalinside}
\widetilde{\Ebf}^0 = \sum_{n=0}^\infty \frac{\gamma(\omega)}{\gamma(\omega) - \gamma_n} \left\langle \widetilde{\Ebf}^\text{\emph{in}}, \ebf_n\right\rangle_{L^2(B,\R^3)} \ebf_n \qquad \tin B.
\end{align}
\end{theorem}

\begin{proof} Theorem~\ref{theo:staticmodal} can be deduced as a direct consequence of Proposition~\ref{prop:eigenbasis} and equation~\eqref{eq:fieldinsidestatic}.
\end{proof}

\subsection{Sum truncation for strictly convex domains}\label{sec:sumtruncation}
In this section we wish to apply perturbation theory tools to express the solutions of~\eqref{eq:rescaled} in terms of the eigenvectors of $\mathcal{T}^0_B$ that appear in the spectral decomposition of the limiting problem in Theorem~\ref{theo:staticmodal}, and to replace $\gamma_n$ by a perturbed value $\gamma_n(\omega\delta c^{-1})$.
Classical perturbation theory will give us a Taylor expansion for $\gamma_n(\omega\delta c^{-1})$ in $\omega\delta c^{-1}$ for any $n\in \N$ but the remainders and validity range of these expansions will depend on the index $n$ of the eigenvalue. In order to get a meaningful expansion of the scattered field we need to consider a finite number of modes.

Using recent results on the principal symbol of the \emph{Neumann-Poincar\'e operator} \cite{miyanishi2018weyl,ando2020surface}, we were able to prove in \cite{baldassari2021modal} the following result:

\begin{proposition}  \label{prop:polyndecay}For $B$, a strictly convex domain in $\R^3$ with $C^\infty$-smooth boundary, $(\phi_n)_{n \in \N}$ the orthonormal eigenfamily of $\mathcal{K}_B^*$ in $\mathcal{H}^*(\partial B)$, and $V^\text{\emph{in}}\in H^N(\partial B)$ for some $N \in \N^*$ we have :\begin{equation*}
	\left\langle V^\text{\emph{in}}, \phi_n\right\rangle_{\mathcal{H}^*(\partial B)}=o(n^{-N/4}) \qquad \text{~as~} n\rightarrow +\infty.
	\end{equation*}
\end{proposition}

\begin{corollary}\label{cor:truncate_vector}
Let $\widetilde{\Ebf}^\text{\emph{in}} \in \mathbf{W}(B)\cap H^{N+1/2}(B,\R^3)$. Then:
\begin{equation*}
 \left\langle \widetilde{\Ebf}^\text{\emph{in}}, \ebf_n\right\rangle_{L^2(B,\R^3)} = o(n^{-N/4}) \qquad \text{~as~} n\rightarrow +\infty.
\end{equation*}
\end{corollary}

\begin{proof}
We start by writing that
\begin{align*}
\left\langle \widetilde{\Ebf}^\text{in}, \ebf_n\right\rangle_{L^2( B,\R^3)}&= \left\langle\widetilde{\Ebf}^\text{in} , (\gamma_n)^{-1}\nabla\mathcal{S}_B[\ebf_n\cdot \nu] \right\rangle_{L^2( B,\R^3)}\\&= (\gamma_n)^{-1}\left\langle \widetilde{\Ebf}^\text{in}\cdot\nu, \mathcal{S}_B\left[ \ebf_n\cdot\nu\right] \right\rangle_{L^2(\partial B)}\\
&=-(\gamma_n)^{-1} \left\langle \widetilde{\Ebf}^\text{in}\cdot\nu, \ebf_n\cdot\nu\right\rangle_{\mathcal{H}^*(\partial B)}.
\end{align*}
Since by construction of the eigenbasis of $\mathcal{T}_B^0$, $(\ebf_n\cdot\nu)_{n\in\N}$ is the orthogonal family (in the sense of $\mathcal{H}^*(\partial B)$) of eigenvectors  of $\mathcal{K}_B^*$, and by hypothesis $\widetilde{\Ebf}^\text{in}\cdot\nu\in H^{N}(\partial B)$ we can apply Proposition~\ref{prop:polyndecay} and get the result.
\end{proof}

The immediate consequence of Corollary~\ref{cor:truncate_vector} is that even with a finite number of modes we get a very accurate approximation of the static electric field:
\begin{corollary}\label{cor:modalapproxincidentfield}For any $N \in \N^*$
\begin{align*}
\left\Vert\widetilde{\Ebf}^0- \sum_{n=0}^{N_0} \frac{\gamma(\omega)}{\gamma(\omega) - \gamma_n} \left\langle \widetilde{\Ebf}^\text{\emph{in}}, \ebf_n\right\rangle_{L^2(B,\R^3)} \ebf_n \right\Vert \leq C_N\left( \frac{N_0^{-N}}{\inf_{n>N_0}  \vert\gamma(\omega)-\gamma_n \vert} \right)\qquad \text{as } N_0\rightarrow +\infty,
\end{align*} where $C_N$ is a constant depending on $N$.
\end{corollary}
\begin{remark}Since the eigenvalues $\gamma_n$ accumulate around $\frac{1}{2}$ as $\frac{1}{2} \pm n^{-1/2}$ \cite{ando2020spectral} the term $\inf_{n>N_0}  \vert\gamma(\omega)-\gamma_n \vert$ is bounded away from zero if $\gamma(\omega)$ is bounded away from $\frac{1}{2}$. In practice, metals have high absorption in the optical frequencies, and $\Im \, \varepsilon_c(\omega)$ is of order one, while the surrounding medium usually has low absorption. So in practical situations, $\Im\, \gamma(\omega)$ is of order one and the term $\inf_{n>N_0}\vert\gamma(\omega)-\gamma_n \vert$ is of order one and can be ignored.
\end{remark}

\section{Dynamic regime}\label{sec:dynamic}
In this section we consider a regime where $\omega \delta c^{-1}$ is small but not zero. Our objective is to solve \eqref{eq:rescaled} to obtain an approximation of the electric field in the particle in this  low-frequency regime.

\subsection{Preliminary results}
\begin{definition}We denote by $\mathcal{R}^{\frac{\omega\delta}{c}}_B(\gamma)$ the resolvent of $\mathcal{T}_B^{\frac{\omega\delta}{c}}$.
\end{definition}

\begin{lemma}\label{lem:perturbation1}The following holds:
\begin{align*}
\left\Vert \mathcal{T}^{\frac{\omega\delta}{c}}_B - \mathcal{T}^0_B\right\Vert_{\mathcal{L}(L^2)} \leq C_B \frac{\delta^2\omega^2}{c^2},
\end{align*}
where $C_B$ is a constant depending only on $B$. Moreover, the perturbation $\mathcal{T}^{\frac{\omega\delta}{c}}_B - \mathcal{T}^0_B$ is a compact operator in $L^2(B,\R^3)$. The essential spectrum of $\mathcal{T}_B^{\frac{\omega\delta}{c}}$ is the same as $\mathcal{T}^0_B$, $\sigma_{ess}\left(\mathcal{T}_B^{\frac{\omega\delta}{c}}\right)= \{0,\frac{1}{2}, 1\}$. The rest of the spectrum of $\mathcal{T}_B^{{\frac{\omega\delta}{c}}}$ is a discrete bounded countable set of eigenvalues.
\end{lemma}

\begin{proof}
We start by writing the asymptotic development of $\mathcal{T}_B^{\frac{\omega\delta}{c}}$. Let $\mathbf{f}\in L^2(B,\R^3)$ and $x\in \R^3$. Then
\begin{align*}
\mathcal{T}^{\frac{\omega\delta}{c}}_B [\mathbf{f}](x) =&-\left(\frac{\omega\delta}{c}\right)^2 \int_B \Gamma^{{\frac{\omega\delta}{c}}}(x,y) \mathbf{f}(y)\dd y - \nabla\int_B \nabla_x \Gamma^{{\frac{\omega\delta}{c}}}(x,y) \cdot \mathbf{f}(y) \dd y\\
= & \left(\frac{\omega\delta}{c}\right)^2 \int_B \frac{\sum_k \left(i \frac{\omega\delta}{c}\vert x-y\vert\right)^k}{ k! 4\pi \vert x-y\vert} \mathbf{f}(y)\dd y + \nabla\int_B \nabla_x \frac{\sum_k \left(i \frac{\omega\delta}{c}\vert x-y\vert\right)^k}{ k! 4\pi \vert x-y\vert} \cdot \mathbf{f}(y) \dd y\
\\
=& \mathcal{T}_B^0[\mathbf{f}](x) + \left(\frac{\omega\delta}{c}\right)^2 \mathcal{T}_B^{(2)}[\mathbf{f}](x) + i \left(\frac{\omega\delta}{c}\right)^3 \mathcal{T}_B^{(3)}[\mathbf{f}](x) + \mathcal{O}\left(\left(\frac{\omega\delta}{c}\right)^4\left\Vert  \mathbf{f}\right\Vert \right),
\end{align*}
with
\begin{align}\label{eq:defT2}
\mathcal{T}_B^{(2)}[\mathbf{f}]= \frac{1}{8\pi}\int_B\left(\Ibf- \frac{(x-y)(x-y)^\top}{\vert x-y\vert^2} \right)\frac{\mathbf{f}(y)}{\vert x - y\vert} \dd y.
\end{align}
Note that $\mathcal{T}_B^{(2)}$ is a compact self-adjoint operator on $L^{2}(B,\R^3)$ and  $\mathcal{T}_B^{(3)}$ is a compact operator.
\end{proof}

Since $\mathcal{T}^{\frac{\omega\delta}{c}}_B - \mathcal{T}^0_B$ is not a self-adjoint perturbation of $\mathcal{T}_B^0$ we only know that the spectrum of $\mathcal{T}^{\frac{\omega\delta}{c}}_B$ is upper semi-continuous with respect to the parameter $\omega\delta c^{-1}$. Having an explicit bound on the location of all the eigenvalues of $\mathcal{T}_B^{\frac{\omega\delta}{c}}$ requires some extra work.
When considering a single (or a finite number) of eigenvalues, it is possible to have an explicit asymptotic formula for the perturbed eigenvalues when the perturbation parameter $\omega\delta c^{-1}$ lies in a complex neighbourhood of the origin:

\begin{lemma}\label{lem:perturbationtheory}Consider the eigenvector $\ebf_n\in \mathbf{W}(B)$ associated with  a simple isolated eigenvalue $\gamma_n$. Denote by $d_n$ the spectral isolation distance of $\gamma_n$, \emph{i.e.} the distance $d(\gamma_n, \sigma(\mathcal{T}_B^0)\setminus\{\gamma_n\})$. Then for $$\eta:= \left\Vert \mathcal{T}^{\frac{\omega\delta}{c}}_B - \mathcal{T}^0_B\right\Vert_{\mathcal{L}(L^2)} \leq d_n,$$ there exists only one eigenvalue $\gamma_{n}(\frac{\omega\delta}{c})$ of $\mathcal{T}_B^{\frac{\omega\delta}{c}}$  in the strip \begin{align*}
\left[\gamma_n - \eta , \gamma_n +\eta \right]+ i\R.
\end{align*}
Moreover \begin{align*}
\gamma_{n}\left(\frac{\omega\delta}{c}\right)= \gamma_n + \left\langle \left(\mathcal{T}^{\frac{\omega\delta}{c}}_B - \mathcal{T}^0_B \right)[\ebf_{n}],\ebf_{n}\right\rangle_{L^2(B,\R^3)}+\mathcal{O}(\delta^3\omega^3 c^{-3}),
\end{align*} and an associated eigenvector can be written as
\begin{align*}
\ebf_n\left(\frac{\omega\delta}{c}\right)= \ebf_n + \sum_{k\neq n} \frac{\left\langle \left(\mathcal{T}^{\frac{\omega\delta}{c}}_B - \mathcal{T}^0_B \right)[\ebf_{n}],\ebf_{k}\right\rangle_{L^2(B,\R^3)}}{\gamma_n-\gamma_k} \ebf_k +\mathcal{O}(\delta^3\omega^3 c^{-3}).
\end{align*}
\end{lemma}

\begin{proof}To apply perturbative theory results, we note that $\mathcal{T}_B^{\frac{\omega\delta}{c}}$ is a bounded perturbation of the self-adjoint operator $\mathcal{T}^0_B$  (Lemma~\ref{lem:perturbation1}).
The result is a consequence of \cite[Theorem 2.12]{cuenin2016non}, with $a= \left\Vert \mathcal{T}^{\frac{\omega\delta}{c}}_B - \mathcal{T}^0_B\right\Vert_{\mathcal{L}(L^2)}$ and $b=0$.
The formulae for the perturbed eigenvalues and eigenvectors are classical and can be found in Kato's book \cite[Chapter 8, section 2.3]{kato2013perturbation} or Reed $\&$ Simon \cite[Chapter XII]{reed1978methods}.

\end{proof}

Even though we can not derive a perturbation formula that is uniformly valid for all eigenvalues, we can obtain some information about the eigenvalues' location using the method developed in \cite{ammari2020superresolution}.
\begin{lemma} \label{lem:controlresolv}
 $\mathcal{A}_I^{\frac{\omega\delta}{c}}= \left(\mathcal{T}_B^{\frac{\omega\delta}{c}} -\left(\mathcal{T}_B^{\frac{\omega\delta}{c}}\right)^*\right)/(2i)$, the imaginary Hermitian component of $\mathcal{T}^{\frac{\omega\delta}{c}}$, is a Hilbert-Schmidt operator and we have the following resolvent estimate:
\begin{align*}
\left\Vert\mathcal{R}^{\frac{\omega\delta}{c}}_B(\gamma)\right\Vert \leq \frac{\sqrt{2}}{d\left(\gamma,\sigma\left(\mathcal{T}_B^{\frac{\omega\delta}{c}}\right)\right)} \exp{\left({\frac{g_I^2\left(\mathcal{T}_B^{\frac{\omega\delta}{c}}\right)}{d\left(\gamma,\sigma\left(\mathcal{T}_B^{\frac{\omega\delta}{c}}\right)\right)^2}}\right)},
\end{align*}
where $$g_I\left(\mathcal{T}_B^{\frac{\omega\delta}{c}}\right)=\left[\left\Vert \mathcal{A}_I^{\frac{\omega\delta}{c}}\right\Vert_{HS}^2 -\sum_n\left( \Im \gamma_n\left(\frac{\omega\delta}{c}\right) \right)^2 \right]^{\frac{1}{2}},$$
and $\Vert \cdot\Vert_{HS}$ denotes the Hilbert-Schmidt norm.
Moreover, let $\gamma_n\left(\frac{\omega\delta}{c}\right)$ be an eigenvalue of $\mathcal{T}_B^{\frac{\omega\delta}{c}}$ with $\omega\delta c^{-1}\in \R$. Then, when $\delta \omega c^{-1}\rightarrow 0$, \begin{align*}
\left\vert \Im \gamma_n\left(\frac{\omega\delta}{c}\right)\right\vert \leq \widetilde{C}_B \left(\frac{\omega\delta}{c}\right)^3
\end{align*} where $\widetilde{C}_B$ \emph{does not depend} on $n$.
\end{lemma}
\begin{proof} The first part of the result is proved in \cite[Appendix B.1]{ammari2020superresolution}, as well as the resolvent estimate which is a direct consequence of \cite[Theorem 7.7.1 p.106]{gil2003operator}. To prove the estimate on the imaginary part of the eigenvalues, recall that (see the proof of Lemma~\ref{lem:perturbation1})
\begin{align*}
\mathcal{T}^{\frac{\omega\delta}{c}}_B  = \mathcal{T}_B^0 + \left(\frac{\omega\delta}{c}\right)^2 \mathcal{T}_B^{(2)}+ i \left(\frac{\omega\delta}{c}\right)^3 \mathcal{T}_B^{(3)}+ \mathcal{O}\left(\left(\frac{\omega\delta}{c}\right)^4 \right),
\end{align*} with $\mathcal{T}_B^{(2)}$ being a compact self-adjoint operator. Therefore
\begin{align*}
\left\Vert\mathcal{A}_I^{\frac{\omega\delta}{c}}\right\Vert \sim \left(\frac{\omega\delta}{c}\right)^3 \left\Vert  \mathcal{T}_B^{(3)}\right\Vert.
\end{align*}
Since $\left\Vert\mathcal{A}_I^{\frac{\omega\delta}{c}}\right\Vert^2= \rho\left( \mathcal{A}_I^{\frac{\omega\delta}{c}} \left(\mathcal{A}_I^{\frac{\omega\delta}{c}}\right)^*\right)$, we immediately get that
 \begin{align*}
\sup_{n\in\N}\left[ \Im \gamma_n\left(\frac{\omega\delta}{c}\right) \right]^2 \leq & \left\Vert\mathcal{A}_I^{\frac{\omega\delta}{c}}\right\Vert^2\\
\leq & \left[\left(\frac{\omega\delta}{c}\right)^3 \left\Vert  \mathcal{T}_B^{(3)}\right\Vert\right]^2.
\end{align*}

\end{proof}

\begin{lemma}[Isolation distance]\begin{align*}
d_n\sim n^{-1/2} \quad \text{as } n\rightarrow +\infty.
\end{align*}
\end{lemma}
\begin{proof}The proof is a direct consequence of $\lambda_n \sim n^{-1/2}$ as $n\rightarrow +\infty$ \cite{miyanishi2020eigenvalues} and $\gamma_n=1/2-\lambda_n$. We then have $\gamma_n - \gamma_{n+1} \sim\pm n^{-1/2}$.
\end{proof}

\begin{lemma}\label{lem:resolvW} Let $\mathbf{F}\in \mathbf{W}(B) \cap \mathcal{C}^\infty(\overline{B}, \R^3)$.
For any $N\in\N$, there exists a  complex neighbourhood $\mathcal{V}(N)\ni 0$ such that for $\omega\delta c^{-1}\in \mathcal{V}(N)$ and $k\in\N$:
\begin{align*}
\mathcal{R}_B^{ \frac{\omega\delta}{c}} (\gamma)[\mathbf{F}] = \sum_{n=0}^N \frac{\left\langle \mathbf{F},\ebf_n\right\rangle_{L^2\left(B,\R^3\right)}}{\gamma-\gamma_n\left(\frac{\omega\delta}{c}\right)} \ebf_n + \mathcal{O}\left( \left\Vert\mathcal{R}^{\frac{\omega\delta}{c}}(\gamma)\right\Vert\left(\frac{\omega\delta}{c}\right)^2\right)  + \mathcal{O}\left( \left\Vert\mathcal{R}^{\frac{\omega\delta}{c}}(\gamma)\right\Vert\epsilon_N(\mathbf{F})\right),
\end{align*} with $\epsilon_N(\mathbf{F})=\mathcal{O}(N^{-k})$.
\end{lemma}

\begin{proof}Since $\mathbf{F}\in \mathbf{W}(B) \cap \mathcal{C}^\infty(\overline{B}, \R^3)$, we have a superpolynomial decay of the coefficients $\left\langle \mathbf{F},\ebf_n\right\rangle_{L^2(B,\R^3)}$ by Lemma~\ref{cor:truncate_vector}, \emph{i.e.} for a fixed $\mathbf{F}$, $\epsilon_N(\mathbf{F}):=\left\Vert  \left(I - \sum_{n=0}^N\mathcal{P}_n^{0}\right)[\mathbf{F}]\right\Vert = \mathcal{O}(N^{-k})$ for any $k\in\N$ where $\mathcal{P}_n^{0}$ is the Riesz projection onto $\ebf_n$.
For a fixed $N$ the perturbation formula is uniformly valid for all eigenvalues $\gamma_n$ with $n<N$ for a perturbation parameter satisfying $\vert \omega^2 \delta^2 c^{-2}\vert< d_N \sim N^{-1/2}$.
If we define $\mathcal{V}_N(\mathbf{F}):= \left\{z\in \C,\, \vert z \vert < d_N\right\}$ then, for $\omega^2 \delta^2 c^{-2} \in \mathcal{V}_N(\mathbf{F})$ and using the continuity of $\mathcal{R}^{\frac{\omega\delta}{c}}(\gamma)$, it follows that

\begin{align*}
\mathcal{R}^{\frac{\omega\delta}{c}}(\gamma) [\mathbf{F}] = \mathcal{R}^{\frac{\omega\delta}{c}}(\gamma)\left[\sum_{n=1}^N \mathcal{P}^{0}_n[\mathbf{F}]\right] +\mathcal{O}\left(\left\Vert\mathcal{R}^{\frac{\omega\delta}{c}}(\gamma)\right\Vert \left(\frac{\omega\delta}{c}\right)^2\right).
\end{align*}
\end{proof}

Contrarily to the static case, $\widetilde{\Ebf}^\text{in}$ does not belong to $\mathbf{W}(B)$. However, $\widetilde{\Ebf}^\text{in}$ can be approximated by elements of $\mathbf{W}(B)$:
\begin{lemma}\label{lem:projection} We have $\left(I - P_\mathbf{W}\right)\left[\widetilde{\Ebf}^\text{\emph{in}}\right]\in \mathbf{H}_0(\text{\emph{div} } 0,B)$ and \begin{align*}
\left\Vert \left(I - P_\mathbf{W}\right)\left[\widetilde{\Ebf}^\text{\emph{in}}\right]\right\Vert_{L^2(B,\R^3)}  = \mathcal{O}\left(\left(\frac{\omega\delta}{c}\right)^2 \right).
\end{align*}
\end{lemma}
\begin{proof}
The fact that $\left(I - P_\mathbf{W}\right)\left[\widetilde{\Ebf}^\text{in}\right]\in \mathbf{H}_0(\text{div}\, 0,B)$ can be deduced from the fact that $\widetilde{\Ebf}^\text{in}$ is a solution of Maxwell's equation in an homogeneous medium and therefore is a divergence-free vector field. The projection of $\widetilde{\Ebf}^\text{in}$ on $\nabla H^1_0(B)$ is the gradient of harmonic function with zero Dirichlet trace on $\partial B$ and is zero.
Write $\widetilde{\Ebf}^\text{in}(\tilde{x})=\Gabf^{\frac{\omega\delta}{c}}(\tilde{x},\tilde{y})\pbf$.
Since $\tilde{x}\mapsto \Gabf^{0}(\tilde{x},\tilde{y}) \in \mathbf{W}(B)$ and $\left\vert \Gabf^{\frac{\omega\delta}{c}}(\tilde{x},\tilde{y}) - \Gabf^{0}(\tilde{x},\tilde{y}) \right\vert \leq C \omega^2\delta^2 c^{-2}$ we get the result.
\end{proof}

\subsection{Modal approximation}

We now state the main result.
\begin{proposition}\label{prop:modalinside}
There exists a sequence with superpolynomial decay $\left(\epsilon_N(\widetilde{\Ebf}^\text{\emph{in}})\right)_{N\in\N}$ depending only on $\widetilde{\Ebf}^\text{\emph{in}}$ and $B$, and a sequence of open complex neighbourhood of the origin $\mathcal{V}(N)\ni 0$ such that for $\omega\delta c^{-1}\in \mathcal{V}(N)\cap\R$ the electric field solution of~\eqref{eq:rescaled} satisfies:
\begin{align}\label{eq:modalinside}
\widetilde{\Ebf} =  \sum_{n=0}^{N} \frac{\left\langle \widetilde{\Ebf}^\text{\emph{in}}, \ebf_n\right\rangle_{L^2(B,\R^3)} }{\gamma(\omega) - \gamma_n(\frac{\omega\delta}{c}) }\ebf_n  + \mathcal{O}\left(\frac{\left(\omega\delta c^{-1}\right)^2}{f\left(\left\vert \Im \gamma(\omega) - \widetilde{C }_B\left(\omega\delta c^{-1}\right)^3\right\vert\right)}\right) +\mathcal{O}\left(\frac{\epsilon_N(\widetilde{\Ebf}^\text{\emph{in}})}{f\left(\left\vert \Im \gamma(\omega) - \widetilde{C }_B\left(\omega\delta c^{-1}\right)^3\right\vert\right)}\right),
\end{align} in $B$ and $f(x):=x e^{x^{-2}}$.
\end{proposition}
\begin{remark} The first error term is due to the error of approximating the incoming field by a function of $\mathbf{W}(B)$ while the second error term is due to the fact that we approximate the projection of $\widetilde{\Ebf}^\text{\emph{in}}$ in $\mathbf{W}(B)$ by a finite number of modes. The spectral perturbative theory is only valid when the perturbation is smaller than the isolation distance of the last eigenvalue $\gamma_N$ to the rest of the spectrum. Therefore $\mathcal{V}(N)$ becomes increasingly smaller as $N\rightarrow \infty$. In practice though, $\epsilon_N(\widetilde{\Ebf}^\text{\emph{in}})$ decays very quickly so only a few modes are necessary to describe the field, and $\mathcal{V}(N)$ is large enough for applications.
\end{remark}

\begin{proof}
We start by decomposing the incoming electric field
\begin{align*}
\widetilde{\Ebf}^\text{in}= P_\mathbf{W} \left[\widetilde{\Ebf}^\text{in} \right] +  \left(I - P_\mathbf{W}\right)\left[\widetilde{\Ebf}^\text{in}\right].
\end{align*}
Using Lemma~\ref{lem:projection} we have $\left(I - P_\mathbf{W}\right)\left[\widetilde{\Ebf}^\text{in}\right]\in \mathbf{H}_0(\text{div}\, 0,B)$  and  \begin{align*}
\left\Vert P_{\mathbf{H}_0}\left[\widetilde{\Ebf}^\text{in}\right]\right\Vert_{L^2(B,\R^3)}  = \mathcal{O}\left(\delta^2 \omega^2 c^{-2}\right).
\end{align*}
Then we can apply Lemma~\ref{lem:resolvW} to the function $\mathbf{F} = P_\mathbf{W}\left[\widetilde{\Ebf}^\text{in} \right]$ and get the result using the resolvent estimate of  Lemma~\ref{lem:controlresolv}.
\end{proof}

It is clear from expression~\eqref{eq:modalinside} that the resolvent has poles in the complex plane at the roots of the equation $\gamma(\omega)=\gamma_n(\omega\delta c^{-1})$.
Before we can study the electric field in the time domain, we need to determine the poles' locations.
In the next section we study the roots of the equations:
\begin{align*}
\gamma(\omega)=&~\gamma_n  \qquad &\text{(static \ regime),}\\
\gamma(\omega)=&~\gamma_n(\omega\delta c^{-1})\qquad &\text{(dynamic \ regime)}.
\end{align*}

\begin{remark}[Equivalent method with boundary integral operators]\label{rem:equivalentscalar} Corollary~\ref{cor:linkk*} shows that there is a strong link between the surface integral operator and the volume integral operator. A similar type of modal expansion can be obtained using layer potential operators. The layer potential operators describing the scattering problem act on $L^2_T(\partial D)$ the space of vector fields in $L^2$ that are tangential to the particle. The vectorial equivalent of the Neumann-Poincar\'e operator that appears cannot be symmetrised as easily as $\mathcal{K}^*_D$ in the scalar case. One has to perform a Helmholtz type decomposition on the $L^2_T(\partial D)$  vector fields, and the symmetrisation is only valid on one of the subspaces, see~\cite[section 4]{ammari2016surface} for more details. The computation of the perturbed spectrum can be carried out as in~\cite{ammari2016plasmaxwell}. Nevertheless, they are quite technical, making the result more difficult to interpret.
\end{remark}

\begin{remark}[About the sum truncation]\label{rem:sumtruncation}
In some cases, it is possible to perform the perturbative analysis for the whole spectrum at once. In \cite[chapter 5, section 4, subsection 5]{kato2013perturbation} there are some results regarding the completeness of eigenprojections for non-symmetric $T$-bounded perturbations of a self-adjoint operator. Nevertheless they are valid for a self-adjoint operator with a compact resolvent and rely on the eigenvalues going to infinity with isolation distance going to infinity. This is not applicable in our case as we have an accumulation point at $\frac{1}{2}$ and an isolation distance vanishing. We can also not apply the result to $L:=\left(\frac{1}{2}I- \mathcal{T}\right)^{-1}$ because it does not have a compact resolvent. 
There are some more recent results regarding the perturbation of the spectrum when the isolation distance vanishes~\cite{barone2016remark} but they apply to symmetric perturbations $T_V=T+V$ of a self-adjoint operator $T$. Moreover, an exponential decay of the scalar product of the quantities $\left\langle V e_n, e_m\right\rangle$ is required, with respect to $m$ and $n$, where $(e_n)_{n\in\N}$ are the eigenvectors of the unperturbed operator.
\end{remark}

\section{Static and dynamic plasmonic resonances}\label{sec:statanddyn}

From the static modal expansion~\eqref{eq:staticmodalinside} and the dynamic modal approximation~\eqref{eq:modalinside}  it is clear that we can define two types of resonances.
\begin{definition}
We say that $\Omega\in \C$ is a \emph{static} resonance if $\gamma(\Omega)\in \sigma\left(\mathcal{T}^0_D\right)$.
We say that $\Omega\in \C$  is a \emph{dynamic} resonance if $\gamma(\Omega)\in \sigma\left(\mathcal{T}_D^{\frac{\Omega}{c}}\right)$.
\end{definition}
In what follows we use the lower-case character $\omega$ for real frequencies and the upper-case character $\Omega$ for complex frequencies.

\subsection{Static plasmonic resonances}
\begin{proposition}\label{prop:QS_res_max}
Assuming $\varepsilon_m\in \R^+$ and using the Drude model~\eqref{eq:drude} the static plasmonic resonances have an explicit formula: $\Omega_n=\Omega'_n+i\Omega''_n$ such that for $n\geq1 $,
\begin{align*}
\Omega'_n&= \pm\sqrt{\dfrac{\omega_p^2}{1-\dfrac{\gamma_n-1}{\gamma_n}\dfrac{\varepsilon_m}{\varepsilon_0}}-\frac{1}{4 \mathrm{T}^2} }, \\
\Omega''_n& =-\frac{1}{2\mathrm{T}}.
\end{align*}
The static plasmonic resonances all lie in the lower part of the complex plane and their real parts are bounded.
\end{proposition}
\begin{proof}
Let $\gamma_n\in \sigma(\mathcal{T}_D^0)$:
\begin{align*}
\frac{\varepsilon_m}{\varepsilon_m-\varepsilon_c(\omega)}=\gamma_n &\Leftrightarrow \varepsilon_c(\omega)=\varepsilon_m\left(1-\frac{1}{\gamma_n}\right) \\
& \Leftrightarrow \frac{\omega_p^2}{\omega^2+i\omega \mathrm{T}^{-1}} = 1-\frac{\gamma_n-1}{\gamma_n}\frac{\varepsilon_m}{\varepsilon_0}\\
&\Leftrightarrow \Omega'^2 -\Omega''^2 +2i\Omega' \Omega'' +i \Omega' \mathrm{T}^{-1} - \Omega'' \mathrm{T}^{-1} = \dfrac{\omega_p^2}{1-\dfrac{\gamma_n-1}{\gamma_n}\dfrac{\varepsilon_m}{\varepsilon_0}},
\end{align*}
which gives the result.
\end{proof}

\subsection{Dynamic plasmonic resonances}
Finding the frequencies $\Omega$ at which a dynamic plasmonic resonance can occur is the non-linear eigenvalue problem of
\begin{align}\label{eq:nonlinearEV}
\text{finding}\ \Omega \ \text{s.t.} \ \gamma(\Omega)\in\sigma\left(\mathcal{T}_D^{\frac{\Omega}{c}}\right).
\end{align}

There are several difficulties associated with this problem. First, when $\delta$ is fixed, there might not be a solution to this problem. Then, even if there is one, we do not know a priori that the solution will correspond to a \emph{low-frequency regime} where our previous perturbative computations can be used.
Assuming without loss of generality, that $\varepsilon_m=\varepsilon_0$, and that $\varepsilon_c$ is given by the Drude model~\eqref{eq:drude},  we show the following result:
\begin{proposition}\label{prop:existanceresonances} For a fixed static eigenvalue $\gamma_n$  we can find $\delta_{max}(n)$ such that if $0<\delta<\delta_{max}(n)$ then there exists a \emph{low-frequency} solution of the non-linear eigenvalue problem, \emph{i.e.}
\begin{align*}
\forall 0<\delta<\delta_{max}(n) \quad \exists\  \Omega_n(\delta)\ \mathrm{such\ that\ } \gamma(\Omega_n(\delta))=\gamma_n(\delta\Omega_n(\delta) c^{-1}) \ \mathrm{ and\ } \delta\omega_n(\delta) c^{-1} \ll 1.
\end{align*}
\end{proposition}
In what follows, we use the lighter notation $\gamma_n(\Omega_n(\delta))$ instead of $
\gamma_n\left(\delta\Omega_n(\delta) c^{-1}\right)$.
\begin{proof}
This is a consequence of Rouch\'e's theorem.
Consider $f(\omega)=\gamma(\omega)-\gamma_n$ and $g(\omega)= \gamma_n- \gamma_n(\omega)$.
Then, we know by Rouch\'e's theorem that if there is a subset $K\subset\C$ such that $\vert g(z)\vert <\vert f(z)\vert$ on $\partial K$ then $f$ and $f+g$ have the same number of zeros in $K$.
Using Proposition~\ref{prop:QS_res_max} we know that there is a neighbourhood $K\ni 0$ of $\Omega_n$ such that $f$ only has one zero in $K$.
Using Lemma~\ref{lem:perturbationtheory} we know that there is a complex neighbourhood of the origin $\mathcal{V}(n)$ such that if $\omega\delta c^{-1}\in \mathcal{V}(n)$ then $g(\omega) \sim -\alpha_n \left(\omega\delta c^{-1}\right)^2$ where $\alpha_n = \left\langle \mathcal{T}^{(2)}_B[\ebf_n],\ebf_n\right\rangle$. Now, we have
\begin{align*}
\gamma'(\omega) = \frac{\left(2\omega+i T^{-1}\right) }{\omega_p^2}\neq 0 \quad \text{if} \ \omega\neq -\frac{i T^{-1}}{2}.
\end{align*}
 Since $g$ converges uniformly to $0$ in $K$ as $\delta\rightarrow 0$ we can find a $\delta_{max}(n)$ such that the hypotheses of Rouch\'e's theorem hold.
\end{proof}
\begin{remark} The consequence of Proposition~\ref{prop:existanceresonances} is that the poles of the finite dimensional approximation of the resolvent all lie in a bounded region of the lower-half complex plane. 
\end{remark}

\subsection{Low-frequency plasmonic resonance expansion}\label{subsec:modalexpansion}

We can now give an approximation of the scattered electric field as a pole expansion:

\begin{theorem}\label{theo:mainexpansion}
 For a given $\Ebf^\text{\emph{in}}$ there exists $N$ (depending on $\Ebf^\text{\emph{in}}$), $\delta_{max}(N)$  such that for all $\delta<\delta_{max}(N)$, there exists $\omega_{max}=\mathcal{O}\left(c \delta^{-1}\right)$  such that for all $\omega\in\R$ satisfying $\vert \omega\vert<\omega_{max}$ the following holds:
 \begin{align*}
 \Ebf =& \sum_{n=0}^N \frac{\gamma(\omega)}{\gamma(\omega)-\gamma_n\left( \Omega_n(\delta) \right)}\left\langle \Ebf^{\text{\emph{in}}}, \ebf_n \right\rangle_{L^2(D,\R^3)} \ebf_n + \epsilon_{int}\left(N,\frac{\omega\delta }{c}, \Im(\gamma)\right) \tin D,
 \end{align*} and
 \begin{multline*}
 \Ebf(x)-\Ebf^{\text{\emph{in}}}(x)= \sum_{n=0}^N  \frac{1}{\gamma(\omega)-\gamma_n\left( \Omega_n(\delta)\right)} \left\langle \Ebf^{\text{\emph{in}}},\ebf_n \right\rangle_{L^2(D,\R^3)} \left(\frac{\omega}{c}\right)^2 \int_D \Gabf^{\frac{\omega}{c}}(x,y)\ebf_n(y) \dd y \\  + \epsilon_{ext}\left(N,\frac{\omega\delta }{c}, \Im(\gamma)\right) \tin \R^3\setminus \overline{D},
 \end{multline*}
 where $(\ebf_n)_{n\in\N}$ is an orthonormal basis of $\mathbf{W}(D)$ for the usual $L^2(D,\R^3)$ scalar product, and $\gamma_n( \Omega_n(\delta))$ are the eigenvalues of $\mathcal{T}_D^{\frac{\omega}{c}}$ at the dynamic plasmonic resonant frequency $\omega= \Omega_n(\delta)$ on $\mathbf{W}$ associated with the eigenvectors $\ebf_n$.
 The interior error terms behave as :
 \begin{align*}
 \epsilon_{int}\left(N,\frac{\omega\delta }{c}, \Im(\gamma)\right) = \mathcal{O}\left(\frac{\left(\omega\delta c^{-1}\right)^2}{f\left(\left\vert \Im \gamma(\omega) - \widetilde{C }_B\left(\omega\delta c^{-1}\right)^3\right\vert\right)}\right) +\mathcal{O}\left(\frac{\epsilon_N(\Ebf^\emph{\text{in}})}{f\left(\left\vert \Im \gamma(\omega) - \widetilde{C }_B\left(\omega\delta c^{-1}\right)^3\right\vert\right)}\right),
\end{align*} where $f(x):=x e^{x^{-2}}$ and $\epsilon_N(\Ebf^\text{\emph{in}})$ has superpolynomial decay. The exterior error terms are the convolutions by the Green's function on $D$ of the interior error terms and behave as:
\begin{align*}
 \epsilon_{ext}\left(N,\frac{ \omega\delta}{c}, \Im(\gamma)\right)  = \mathcal{O}\left(\frac{\omega^2\delta^2}{c^2}\, \epsilon_{int}\left(N,\frac{\omega\delta }{c}, \Im(\gamma)\right) \right).
\end{align*}
\end{theorem}

\begin{remark}\label{rem:defQNM}
One can see that the function $\omega \mapsto \Ebf(\omega)$ is meromorphic and has simple poles at $\omega = \Omega_n(\delta)$. One can define the so-called \emph{quasi-normal modes} from the physics literature as the excitation independent part of the residue of $\Ebf^{\text{\emph{sca}}}$ at each pole:
\begin{align}
\Ebf_n(x) = \left\{\begin{aligned}\ebf_n(x), \qquad    & x\in D ,\\
\left(\frac{\Omega_n(\delta)}{c}\right)^2 \int_D \Gabf^{\frac{\Omega_n(\delta)}{c}}(x,y)\ebf_n(y) \dd y, \qquad  & x\in \R^3\setminus \overline{D}.
\end{aligned}\right.
\end{align}
Since $\Omega_n(\delta)$ has a negative imaginary part the quasi-normal modes do not belong to $L^2$ and diverge exponentially as $\vert x \vert \rightarrow \infty$.

\end{remark}

\begin{remark} Even though it might be possible, with some extra work, to perform the perturbative spectral analysis without truncating the series (see Remark~\ref{rem:sumtruncation}), studying the excitation coefficients' decay sheds light on the numerical method convergence rate. Indeed, in practice, one has to use a finite number of modes to numerically approximate the electric field. Our result shows that the truncation is valid, and that the number of modes to consider depends on the regularity of the source, hence the convergence is not uniform. For example, a source with high spatial oscillations, like a dipole source placed near the particle, will require a large number of modes to obtain an accurate approximation of the electric field. This has been observed numerically, but not justified. Our result shows that the so-called \emph{quasi-normal modes expansion}'s pertinence as a numerical method depends on the type of source used. 
\end{remark}
We now give the proof of Theorem~\ref{theo:mainexpansion}:
\begin{proof}
The theorem is a direct consequence of Propositions~\ref{prop:modalinside} and~\ref{prop:existanceresonances}. First, we rescale the problem as in section \ref{sec:rescale} to get a problem on $B$ with frequency parameter $\omega\delta c^{-1}$. The next step consists in fixing the incident field as well as the number $N$ of modes to consider in the decomposition (see section \ref{sec:sumtruncation} for more details on the decay of the modal coefficients). In turn the number $N$ fixes the smallest isolation distance of the static operator $\mathcal{T}_B^{0}$ eigenvalues and therefore the maximum size $\eta_{max}$ of the parameter $\omega\delta c^{-1}$ that can be chosen. Then, we compute the dynamic plasmonic resonances (Proposition~\ref{prop:existanceresonances}). The maximum size $\delta_{max}$ is fixed a posteriori  such that for each $n\leq N$ the quantity $\delta_{max} \Omega_n(\delta_{max}) c^{-1} <\eta_{max}$ and that our theory is self-consistent. Then, for any $\delta < \delta_{max}$    the expansion is valid for all $\omega\in \R$ such that $\vert \omega \delta c^{-1}\vert <\eta_{max}$. Therefore setting $\omega_{max} = \eta_{max}\, c \delta^{-1}$ we get a valid expansion in $B$. Going back to the original unscaled problem we get the result. The expansion outside the particle is just the continuation of the field by the Lippman-Schwinger equation.
\end{proof}

\begin{remark} As pointed out in Remark~\ref{rem:equivalentscalar}, there is an equivalent method to obtain a similar type of expansion for the scalar wave equation. In \cite{baldassari2021modal}, simulations are performed to demonstrate the numerical pertinence of such methods. 
\end{remark}

In the next section we show that, in the time domain, the scattered field can be written as a resonance expansion without any divergence problems.

\section{Time domain approximation}\label{sec:temporal}

\subsection{Main result and alternative formulations}
Given a wideband signal $f:t \mapsto f(t) \in C_0^{\infty}([0,C_1])$, for $C_1>0$, we want to express the time domain response of the electric field to an oscillating dipole placed at a source point $s$. We assume that most of the energy of the excitation is concentrated in the low frequencies (i.e., in frequencies corresponding to wavelengths that are much larger than the particle, such that the response of the particle can be studied via the perturbed quasi-static theory).
This means that for a fixed $\delta$ we can pick $\eta \ll 1$ and $\rho$ such that
\begin{align*}
\int_{\R \setminus[-\rho,\rho]} \vert \hf(\omega) \vert^2 \dd \omega \leq \eta,  \\
\frac{\rho \delta}{c} \leq 1,
\end{align*}
where $\hf:\omega\mapsto \hf(\omega)$ is the Fourier transform of $f$.
The goal of this section is to establish a resonance-type expansion for the low-frequency part of the scattered electric field in the time domain, based on the modal approximation established in Theorem~\ref{theo:mainexpansion}.
Introduce, for $\rho>0$,  the truncated inverse Fourier transform of the scattered field $\Ebf^{\text{sca}}$ given by
\begin{equation*}
P_\rho\left[\Ebf^\text{sca}\right](x,t)=\int_{-\rho}^{\rho} \Ebf^{\text{sca}}(x,\omega) e^{-i\omega t} \mathrm{d}\omega.
\end{equation*}

Recall that $z$ is the centre of the resonator and $\delta$ its radius. Let us define $$t_0^\pm(s,x):=\frac{1}{c}\left(|s-z|+|x-z|\pm 2\delta\right),$$
the time it takes to the signal to reach first the scatterer and then observation point $x$. The term $\pm 2 \delta/c$ accounts for the maximal timespan spent inside the particle.

Since the source point $s$ of the incoming field is fixed and, using Theorem~\ref{theo:mainexpansion} we know that there exist $N$ (depending on $\Ebf^\text{in}$) and $\delta_{max}(N)$  such that for all  particle of size $\delta<\delta_{max}(N)$ and  $\vert \omega\vert<\omega_{max}(\delta)$ we can approximate the scattered field by a finite sum:
 \begin{align*}
  \widehat{\Ebf}(\cdot,\omega) =& \sum_{n=0}^N \frac{\gamma(\omega)}{\gamma(\omega)-\gamma_n\left( \Omega_n(\delta) \right)}\left\langle \widehat{\Ebf}^{\text{in}}, \ebf_n \right\rangle_{L^2(D,\R^3)} \ebf_n + \epsilon_{int}\left(N,\frac{\delta \omega}{c}, \Im(\gamma)\right) \tin D,
 \end{align*} where $\widehat{\Ebf}^\text{in}$ is the Fourier transform of $\Ebf^\text{in}$.
The next theorem gives the time domain formulation of the finite sum:
\begin{theorem}\label{theo:resonanceexpansion}
Let $M\in \N^*$. For a particle of size $\delta\leq \delta_{max}$, assume the field has the following form inside the particle in the frequency domain:
\begin{align*}
\widehat{\Ebf}(x,\omega)= \sum_{n=0}^N \frac{\gamma(\omega)}{\gamma(\omega)-\gamma_n\left( \Omega_n(\delta) \right)}\left\langle \widehat{\Ebf}^{\text{\emph{in}}}, \ebf_n \right\rangle_{L^2(D,\R^3)} \ebf_n.
\end{align*} Then, in the time domain, the truncated inverse Fourier transform has the following form, for $x\in\R^3\setminus \overline{D}$:
\begin{equation}
P_\rho\left[\Ebf^\text{\emph{sca}}\right](x,t)=
\begin{dcases}
\mathcal{O}\left(\left(\omega_{max}(\delta)\right)^{-M}\right), &\mbox{ } t\leq t_0^-, \\
2\pi i\sum_{n=1}^NC_{\Omega_n(\delta)} \left\langle\widehat{\Ebf}^{\text{\emph{in}}}\left(\Omega_n(\delta)\right), \ebf_n \right\rangle_{L^2(D,\R^3)}\Ebf_n (x)e^{-i\Omega_n(\delta) t}+\mathcal{O}\left(\frac{1}{t}\left(\omega_{max}(\delta)\right)^{-M}\right), &\mbox{ }t\geq t_0^+,
\end{dcases}
\end{equation}
with $\Omega_n(\delta)$ being the plasmonic resonant frequencies of the particle given by Proposition~\ref{prop:existanceresonances} and $\omega_{max}(\delta)=\mathcal{O}(c \delta^{-1})$ given by Theorem~\ref{theo:mainexpansion},
\begin{align*}
C_{\Omega_n(\delta)}= \text{Res} \left( \frac{\gamma(\omega)}{\gamma(\omega)-\gamma_n(\Omega_n(\delta))},\ \Omega_n(\delta) \right),
\end{align*} and $\Ebf_n$ the generalised (diverging) eigenvectors, the so-called \emph{quasi normal modes} defined in Remark~\ref{rem:defQNM}.
\end{theorem}

\begin{remark}
The resonant frequencies $\Omega_n(\delta)$ have negative imaginary parts, so Theorem~\ref{theo:resonanceexpansion} expresses the scattered field as the sum of decaying oscillating fields. The imaginary part of $\Omega_n(\delta)$ accounts for absorption losses in the particle as well as radiative losses.
\end{remark}

\begin{remark} Even though Theorem~\ref{theo:resonanceexpansion} resembles a resonance expansion similar to the ones found in classical scattering theory, it is not one. It is an approximation of the electric field by a finite number of modes. The number of modes depends on the source and is only valid when the particle is small enough (the maximum size depending on the number of modes). Our resolvent estimates are not uniform with respect to the right-hand side of our equation: They only converge pointwise. Theorem~\ref{theo:resonanceexpansion} provides a good heuristic justification for the methods used in the physics community in a certain regime. However it shows that these modal approximations do not converge uniformly with respect to the source and it seems that getting explicit a priori error estimates will be difficult. 
\end{remark}

\begin{remark}[About the remainder $\omega_{max}(\delta)$]
Since for a particle of finite size $\delta$ our expansion only holds for a range of frequencies $\vert \omega\vert \leq \omega_{max}(\delta)$, we cannot compute the full inverse Fourier transform and we have a remainder that depends on the maximum frequency that we can use. Since that maximum frequency $\omega_{max}(\delta)$ behaves as $c\delta^{-1}$ we can see that the remainder gets arbitrarily small for small particles.  For a completely point-like particle one would get a zero remainder.
\end{remark}

\begin{remark} If we had access to the full inverse Fourier transform of the field, of course, since the inverse Fourier transform of a function which is analytic in the upper-half plane is \emph{causal} we would find that in the case $t\leq \left(|s-z|+|x-z|-2\delta\right)/c$, $\Ebf^\text{\emph{sca}}(x,t)= 0$. Nevertheless, our method only works for a truncated \emph{low-frequency} estimate of the scattered field, hence the \emph{arbitrarily small} remainder.
\end{remark}

\begin{theorem}[Alternative formulation with non-diverging quantities]\label{theo:alternativenondiverging}

Under the same assumptions as Theorem~\ref{theo:resonanceexpansion}, the scattered field has the following form in the time domain for $x\in\R^3\setminus \overline{D}$:
\begin{equation*}
P_\rho\left[\Ebf^\text{sca}\right](x,t)=
\begin{dcases}
\mathcal{O}\left(\left(\omega_{max}(\delta)\right)^{-M}\right), &\mbox{ } t\leq t_0^-, \\
2\pi i\sum_{n=1}^N C_{\Omega_n(\delta)} \left\langle\widehat{\Ebf}^{\mathrm{in}}\left( \Omega_n(\delta) \right), \ebf_n \right\rangle_{L^2(D,\R^3)}\widetilde{\ebf}_n (x)e^{-i\Omega_n(\delta) (t-c^{-1}\vert x-z\vert )}\\  \qquad \qquad \qquad \qquad \qquad \qquad\qquad\qquad\qquad\qquad +\mathcal{O}\left(\frac{1}{t}\left(\omega_{max}(\delta)\right)^{-M}\right), &\mbox{ }t\geq t_0^+,
\end{dcases}
\end{equation*} where
\begin{align*}
\widetilde{\ebf}_n = \Ebf_n e^{-i\frac{\Omega_n(\delta)}{c} \vert x-z\vert }.
\end{align*}
\end{theorem}
\begin{remark}[About the numerical efficiency] One of the main goals of the development of the quasi-normal mode theory in nano-photonics is to compute quickly the scattered field in the time domain in numerical simulations. The idea is that the quasi-normal modes are excitation independent and therefore, once they are pre-computed, one can calculate the scattered field for different source locations very efficiently (by computing a series of scalar products between the source and the modes), compared to the costly time domain finite difference method, which has to be re-done completely if the excitation field is changed.
Theorem~\ref{theo:resonanceexpansion} expresses the scattered field in the time domain as a sum of time decaying exponential times some generalised eigenmodes. Theorem~\ref{theo:alternativenondiverging} says exactly the same thing except that if we make use of the causality, we can express the scattered field as a sum of time decaying exponential functions times a pre-computable, non-diverging quantity. The quantities $\widetilde{\ebf}_n$ are not exactly modes as they do not satisfy Maxwell's equations at frequency $\Omega_n(\delta)$ but they are non-diverging, and they are the quantity that appears \emph{in fine} when computing the residues in practice.
\end{remark}

\subsection{Proof of Theorem~\ref{theo:resonanceexpansion}}
Before we start the proof we need the following lemma:
\begin{lemma}
The incident field has the following form in the time domain:
\begin{align*}
\Ebf^\text{\emph{in}}(x,t)&=\int_{\mathbb{R}} \Gabf^{\frac{\omega}{c}}(x,s)\pbf \hf(\omega)e^{-i\omega t} \mathrm{d}\omega\\
&=\frac{f(t-|x-s|/c)}{4\pi|x-s|}\pbf+c^2\Dbf_x^2\frac{f''(t-|x-s|/c)}{4\pi|x-s|} \pbf.
\end{align*}
\end{lemma}
\begin{proof}
See Appendix \ref{app:timedomain}.
\end{proof}
As well as
\begin{lemma} \label{lem:A} \begin{align*}
\Gabf^\frac{\omega}{c}(x,z)= - e^{i\frac{\omega}{c}\vert x-z\vert}  \frac{\mathbf{A}(x,z,\omega/c)}{4\pi \vert x-z\vert},
\end{align*} where $\mathbf{A}$ is given in Appendix \ref{app:dyadic}, and behaves like a polynomial in $\omega$.
\end{lemma}
We are now ready to prove Theorem~\ref{theo:resonanceexpansion}:
\begin{proof}

We start by studying the time domain response of a single mode to a causal excitation at the source point $s$. Therefore, according to Theorem~\ref{theo:mainexpansion} we need to compute the contribution $\Xi_n$ of each mode $\ebf_n$, that is,
\begin{multline*}
\int_{-\rho}^{\rho}\Xi_n(x,\omega) e^{-i\omega t}\dd \omega \\ :=\int_{-\rho}^\rho\left[\frac{1}{\gamma(\omega)-\gamma_n(\Omega_n(\delta))} \left\langle \Gabf^{\frac{\omega}{c}}(\cdot ,s)\pbf \hf(\omega),\ebf_n \right\rangle_{L^2} \left(\frac{\omega}{c}\right)^2 \int_D \Gabf^{\frac{\omega}{c}}(x,y)\ebf_n(y) \dd y \right] e^{-i\omega t}\dd \omega.
\end{multline*}

One can then write:
\begin{multline*}
\left\langle  \Gabf^{\frac{\omega}{c}}(\cdot ,s)\pbf \hf(\omega),\ebf_n \right\rangle_{L^2} \left(\frac{\omega}{c}\right)^2 \int_D \Gabf^{\frac{\omega}{c}}(x,y)\ebf_n(y) \dd y= \\ \left(\frac{\omega}{c}\right)^2 \hf(\omega) \int_{D\times D} e^{i \frac{\omega}{c} \left( \vert x-y\vert + \vert s-v\vert \right)} \frac{\mathbf{A}(x,y,\omega/c)\ebf_n(y)}{4\pi \vert x-y\vert} \ebf_n(v)\cdot \frac{\mathbf{A}(s,v,\omega/c) \mathbf{p}}{4\pi \vert s-v\vert} \dd v \dd y.
\end{multline*}

Now we want to apply the residue theorem to get an asymptotic expansion in the time domain. Note that:
\begin{equation*}
\int_{-\rho}^{\rho} \Xi_n(x,\omega) e^{-i\omega t} \mathrm{d}\omega= \oint_{\mathcal{C}^{\pm}} \Xi_n(x,\Omega) e^{-i\Omega t}\mathrm{d}\Omega -\int_{\mathcal{C}_\rho^{\pm}} \Xi_n(x,\Omega) e^{-i\Omega t}  \mathrm{d}\Omega,
\end{equation*}
where the integration contour $\mathcal{C}_\rho^{\pm}$ is a semicircular arc of radius $\rho$ in the upper (+) or lower (-) half-plane, and $\mathcal{C}^{\pm}$ is the closed contour $\mathcal{C}^{\pm}=\mathcal{C}_\rho^{\pm}\cup[-\rho,\rho]$.
The integral on the closed contour is the main contribution to the scattered field by the  mode $\ebf_n$ and can be computed using the residue theorem to get, for $\rho\geq \Re[\Omega_n(\delta)]$,
\begin{align*}
\oint_{\mathcal{C}^{+}}  \Xi_n(x,\Omega) e^{-i\Omega t} \mathrm{d}\Omega&=0,\\
\oint_{\mathcal{C}^{-}}  \Xi_n(x,\Omega) e^{-i\Omega t} \mathrm{d}\Omega&=2\pi i\text{Res}\left( \Xi_n(x,\Omega)e^{-i\Omega t},\Omega_n(\delta)\right).
\end{align*}
Since $\Omega_n(\delta)$ is a simple pole of $\omega \mapsto \dfrac{\gamma(\omega)}{\gamma(\omega)-\gamma_n(\Omega_n(\delta))}$ we can write:
\begin{align*}
\oint_{\mathcal{C}^{-}} \Xi_n(x,\Omega) e^{-i\Omega t} \mathrm{d}\Omega&=2\pi i\text{Res}\left(\Xi_n(x,\Omega),\Omega_n(\delta)\right)e^{-i\Omega_n(\delta)t}.
\end{align*}
To compute the integrals on the semi-circle, we introduce:
\begin{equation*}
\Bbf_n(y,v,\Omega)=\frac{\Omega^2}{\gamma(\Omega)-\gamma_n(\Omega_n(\delta))}\frac{\Abf(x,y,\Omega/c)\ebf_n(y) \ebf_n(v)\cdot\Abf(s,v,\Omega/c)\mathbf{p}}{16c^2\pi^2|x-y||s-v|} \qquad (y,v) \in D^2.
\end{equation*}

Note that $\Bbf_n(\cdot,\cdot,\Omega)$ behaves like a polynomial in $\Omega$ when $\vert \Omega\vert \rightarrow \infty$.
Given the regularity of the input signal $f \in C_0^{\infty}([0,C_1])$, the Paley-Wiener theorem~\cite[p.161]{Yosida1995FA} ensures decay properties of its Fourier transform at infinity. For all $M\in\mathbb{N}^*$ there exists a positive constant $C_M$ such that for all $\Omega \in \mathbb{C}$
\begin{equation*}
|\hf(\Omega)|\leq C_M (1+|\Omega|)^{-M}e^{C_1 |\Im{(\Omega)}|}.
\end{equation*}
We now re-write the integrals on the semi-circle
\begin{align*}
\int_{\mathcal{C}_\rho^{\pm}} \Xi_n(x,\Omega) e^{-i\Omega t} \mathrm{d}\Omega=\int_{\mathcal{C}_\rho^{\pm}}\hf(\Omega) \int_{D\times D} \Bbf_n(y,v,\Omega) e^{i \Omega \left(\frac{\vert x-y\vert + \vert s-v\vert}{c} -t\right)}\dd v\dd y \dd \Omega.
\end{align*}
Two cases arise:
\paragraph{Case 1:}
For $0<t<t_0^-$ , i.e., when the signal emitted at $s$ has not reached the observation point $x$, we choose the upper-half integration contour $\mathcal{C}^+$. Transforming into polar coordinates, $\Omega=\rho e^{i\theta}$ for $\theta \in [0,\pi]$, we get:
\begin{align*}
\left\vert  e^{i \Omega \left(\frac{\vert x-y\vert + \vert s-v\vert}{c} -t\right)}\right\vert \leq  e^{-(t_0^--t)\Im(\Omega)} \qquad \forall (y,v)\in D^2,
\end{align*}
and
\begin{align*}
\left|\int_{\mathcal{C}_\rho^{+}}  \Xi_n(x,\Omega) e^{-i\Omega t} \mathrm{d}\Omega\right| & \leq \int_0^\pi\rho \left| \hf \left(\rho e^{i\theta}\right)\right|e^{-\rho (t_0^--t)\sin{\theta}}\int_{D\times D}\left\vert \Bbf_n\left(y,v,\rho e^{i\theta}\right)\right \vert\dd v \dd y \mathrm{d}\theta,\\
&\leq \rho C_M(1+\rho)^{-M}\delta^6 \max_{\theta\in [0,\pi]}{\left\Vert \Bbf_n\left(\cdot, \cdot, \rho e^{i\theta}\right) \right\Vert_{L^{\infty}(D\times D)}} \pi \frac{1-e^{\rho[C_1-(t^-_0-t)]}}{\rho(t^-_0-t-C_1)},
\end{align*}
where we used that for $\theta \in [0,\pi/2]$, we have $\sin{\theta} \geq 2\theta/\pi \geq 0$ and $-\cos{\theta}\leq-1+2\theta/\pi$. The usual way to go forward from here is to take the limit $\rho \rightarrow \infty$, and get that the limit of the integral on the semi-circle is zero. However, we work in the quasi-static approximation here, and our modal expansion is not uniformly valid for all frequencies. So we have to work with a fixed maximum frequency $\rho$. However, the maximum frequency $\rho$ depends on the size of the particle via the hypothesis $\rho \leq \omega_{max}(\delta) $. Since $M$ can be taken arbitrarily large and that $\Bbf_n$ behaves like a polynomial in $\rho$ \emph{whose degree does not depend on $n$}, we get that, uniformly in $n\in [1,N]$:
\begin{align*}
\left|\int_{\mathcal{C}_\rho^{+}} \Xi_n(x,\Omega) e^{-i\Omega t} \mathrm{d}\Omega\right| = \mathcal{O}\left(\frac{1}{t_0^- -t-C_1}\delta^6\left( \omega_{max}(\delta)\right)^{-M}\right).
\end{align*}
Of course if one has to consider the full inverse Fourier transform of the scattered electromagnetic field, by causality, one should expect the limit to be zero. However, one would need high-frequency estimates of the electromagnetic field, as well as a modal decomposition that is uniformly valid for all frequencies.
Since our modal expansion is only valid for a limited range of frequencies we get an error bound that is arbitrarily small if the particle is arbitrarily small, but not rigorously zero.

\paragraph{Case 2:}
For $t>t_0^+$ , we choose the lower-half integration contour $\mathcal{C}^-$. Transforming into polar coordinates, $\Omega=\rho e^{i\theta}$ for $\theta \in [\pi,2\pi]$, we get
\begin{align*}
\left\vert  e^{i \Omega \left(\frac{\vert x-y\vert + \vert s-v\vert}{c} -t\right)}\right\vert \leq  e^{ (t-t_0^+) \Im (\Omega)} \qquad \forall (y,v)\in D^2,
\end{align*}
and
\begin{align*}
\left|\int_{\mathcal{C}_\rho^{-}} \Xi_n(x,\Omega) e^{-i\Omega t} \mathrm{d}\Omega\right| & \leq \int_\pi^{2\pi} \rho\left|f(\rho e^{i\theta})\right|e^{\rho (t-t_0^+)\sin{\theta}}\int_{D\times D}\left\vert \Bbf_n\left(y,v,\rho e^{i\theta}\right) \right\vert\dd v \dd y \mathrm{d}\theta,\\
&\leq \rho C_M(1+\rho)^{-M} \delta^6 \max_{\theta\in [0,\pi]}{\left\Vert \Bbf_n\left(\cdot, \cdot, \rho e^{i\theta}\right) \right\Vert_{L^{\infty}(D\times D)}}\pi \frac{1-e^{\rho( C_1-(t-t_0^+))}}{\rho( C_1-(t-t_0^+))},-.
\end{align*}
Exactly as in Case $1$, we cannot take the limit $\rho \rightarrow \infty$. However, the maximum frequency $\rho$ depends on the size of the particle via the hypothesis $\rho \leq \omega_{max}(\delta) $. Since $\omega_{max}(\delta) \rightarrow \infty$ $(\delta \rightarrow 0)$, using the fact that $M$ can be taken arbitrarily large and that $\Bbf_n$ behaves like a polynomial in $\rho$ \emph{whose degree does not depend on $n$}, we get that, uniformly in $n\in [1,N]$:
\begin{align*}
\left|\int_{\mathcal{C}_\rho^{-}} \Xi_n(x,\Omega) e^{-i\Omega t} \mathrm{d}\Omega\right| = \mathcal{O}\left(\frac{1}{t} \delta^6\left( \omega_{max}(\delta)\right)^{-M}\right).
\end{align*}
The result of Theorem~\ref{theo:resonanceexpansion} is obtained by summing the contribution of all the modes.
\end{proof}

\section{Concluding remarks}
In this paper, we have shown through the spectral analysis of singular integral operators of the Calder\'on-Zygmund type, that the electromagnetic field scattered by a small particle constituted of dispersive media could be approximated in an orthogonal basis inside the particle in the electrostatic regime ($\omega=0$). Through a perturbative analysis of the integral operator on well-chosen finite-dimensional subspaces of $L^2$ we were able to derive an accurate approximation of the electric field in the \emph{low-frequency} regime. The analysis of plasmonic resonances as a non-linear eigenvalue problem of a perturbed operator gives us a \emph{resonance-like} approximation for the electric field : the field is approached by a meromorphic function of the frequency whose poles are all located in a bounded region of the lower-half complex plane.
We have shown that the so-called \emph{quasi-normal modes} that appear in the physics literature~\cite{lalanneReview} can be defined from this expansion but this concept is not the only way for deriving a modal decomposition~\cite{durufle2019non}.
We have then given an approximate resonance expansion in the time domain for the low-frequency part of the scattered electromagnetic field, as a sum of complex exponential (decaying in time) fields, using only results from surface integral operator theory.  We have also shown that there is no divergence problem at infinity once we are in the time domain, and have introduced excitation independent, non divergent quantities that could be used for numerical computations.

\appendix

\section{Fundamental solutions}

\subsection{Green's function}\label{app:def_gamma}
\begin{definition}
Denote by $\Gamma^k$ the outgoing Green's function for the homogeneous medium, i.e., the unique solution of the Helmholtz operator:
\begin{equation*}\label{eq:defGomega}
\left(\Delta +k^2\right)\Gamma^k(\cdot,y)=\delta_y(\cdot) \quad \tin \R^3
\end{equation*} satisfying the Sommerfeld radiation condition. In dimension three, $\Gamma^k$ is given by
\begin{equation*}
\Gamma^k(x,y)=-\frac{e^{ik|x-y|}}{4\pi|x-y|}, \qquad x,y \in \mathbb{R}^3.
\end{equation*}
\end{definition}

\subsection{Dyadic Green's function} \label{app:dyadic}
\begin{definition}
Using the scalar function $\Gamma^k$ defined in Appendix \ref{app:def_gamma} as the fundamental solution to the Helmholtz equation, we now define the matrix-valued function, referred to as the Dyadic Green's function, as
\begin{equation}
\Gabf^{k} (x,y)= -\Gamma^{k} (x,y) \Ibf -\frac{1}{{k}^2} \Dbf_x^2 \Gamma^{k}(x,y), \qquad x,y\in \R^3,
\end{equation}
where $\Ibf$ is the $3\times 3$ identity matrix and $\Dbf_x^2 $ denotes the Hessian.
\end{definition}

\begin{proposition}
$\Gabf^{k_m}$ is a Green's function for the background electric problem, i.e., it satisfies:
\begin{equation*}
\nabla \times \nabla \times \Gabf^{k_m} -{k_m}^2 \Gabf^{k_m} =\delta_y \Ibf \qquad \tin \R^3.
\end{equation*}
\end{proposition}

\begin{lemma}
The matrix $\Abf(x,z,\omega)=(A)_{p,q=1}^3$ introduced in Lemma~\ref{lem:A} is with entries
\begin{align*}
A_{pp}=&\frac{1}{\omega^2|x-z|^4}\Big[-3(x_p-z_p)^2+|x-z|^2+3i\omega(x_p-z_p)^2|x-z|+\omega^2 (x_p-z_p)^2|x-z|^2\\
&-i\omega|x-z|^3-\omega^2|x-z|^4\Big], \\
A_{pq}=&\frac{1}{|x-z|^4}(x_p-z_p)(x_q-z_q)\left[-3+3i\omega|x-z|+\omega^2|x-z|^2\right], \qquad \text{for~} p\neq q.
\end{align*}
\end{lemma}
\subsection{In the time domain}\label{app:timedomain}
In this subsection we compute the inverse Fourier transform of the Green's function.
For a source located at $s \in \R^3$:
\begin{align*}
\Ebf^{\text{in}}(x,t)=&\int_{\R} \Gabf ^\frac{\omega}{c} (x,s)\pbf \hf(\omega) e^{-i\omega t} \dd \omega \\
\nm
=& -\int_{\R} \Gamma ^\frac{\omega}{c} (x,s)\pbf \hf(\omega)e^{-i\omega t} \dd \omega - \int_{\R} \frac{c^2}{\omega^2}\Dbf_x^2 \Gamma^\frac{\omega}{c} (x,s)\pbf f(\omega)e^{-i\omega t} \dd \omega \\
\nm
=& \int_{\R} \dfrac{e^{i\omega|x-s|/c}}{4\pi|x-s|}  \hf(\omega)e^{-i\omega t} \dd \omega \pbf +c^2 \Dbf_x^2  \int_{\R} \dfrac{e^{i\omega|x-s|/c}}{4\pi|x-s|}\frac{ \hf(\omega)}{\omega^2}e^{-i\omega t} \dd \omega \pbf\\
\nm
= & \int_{\R} \frac{e^{-i\omega[t-|x-s|/c]}}{4\pi|x-s|}  \hf(\omega) \dd \omega \pbf  +  c^2 \Dbf_x^2 \int_{\R} \frac{e^{-i\omega[t-|x-s|/c]}}{4\pi|x-s|}  \frac{ \hf(\omega)}{\omega^2} \dd \omega \pbf  \\
\nm
= & \frac{f(t-|x-s|/c)}{4\pi|x-s|}\pbf+c^2 \Dbf_x^2 \frac{f''(t-|x-s|/c)}{4\pi|x-s|} \pbf ,
\end{align*}
where $f''$ is the second derivative of $f$. Note that $f$ and $f''$ vanish for negative arguments, which is physically meaningful since for $t<|x-s|/c$ the direct signal has not reached the observation point yet.

\section{Definitions and properties of some integral operators}

\label{sec:app_def}

We start by defining a singular integral operator, sometimes known as the \emph{magnetisation integral operator}~\cite{friedman1984spectral}.
\begin{definition}\label{de:defTk} Introduce
\begin{align*}
\mathcal{T}_D^k: \begin{aligned} L^2(D,\R^3)& \longrightarrow L^2(D,\R^3) \\
\mathbf{f} &\longmapsto -k^2 \int_{D} \Gamma^k(\cdot,y) \fbf(y)\dd y -\nabla \nabla\cdot \int_D \Gamma^k (\cdot,y) \mathbf{f}(s) \dd y .
\end{aligned}
\end{align*}
\end{definition}

We also need to define the classical \emph{single-layer potential operator}:
\begin{definition}
\begin{align*}
\mathcal{S}_D^k: \begin{aligned} L^2(\partial D)& \longrightarrow L^2(\partial D) \\
\varphi &\longmapsto \int_{\partial D} \Gamma^k(\cdot,y) \varphi(y)\dd \sigma(y),
\end{aligned}
\end{align*} where $\sigma$ is the Lebesgue measure on $\partial D$.
\end{definition}
We also recall the definition of the \emph{Neumann-Poincar\'e operator}:
\begin{definition}
\begin{align*}
\mathcal{K}_D^{k,*}: \begin{aligned} L^2(\partial D)& \longrightarrow L^2(\partial D) \\
\phi &\longmapsto \int_{\partial D} \frac{\partial \Gamma^k(x,y)}{\partial \nu(x)}\phi(y)\mathrm{d}\sigma(y), \qquad x\in \partial D.
\end{aligned}
\end{align*}
\end{definition}
When $k=0$, we just write $\mathcal{S}_D$ and $\mathcal{K}_D^*$ for simplicity.
The following lemmas can be found in~\cite[Chapter 2]{ammari2013mathematical}:
\begin{lemma}\label{lem:unitary} The single-layer potential $\mathcal{S}_D$ is a unitary operator in $H^{-1/2}(\partial D)$ in three dimensions.
\end{lemma}

\begin{lemma}{Calder\'on identity}\label{lem:calderon}
\begin{align*}
\mathcal{S}_D \mathcal{K}_D^* = \mathcal{K}_D \mathcal{S}_D \qquad \tin H^{-1/2}(\partial D).
\end{align*}
\end{lemma}
For more on the symmetrisation property see also~\cite{howland1968symmetrizing}.

\bibliographystyle{siam}
\bibliography{../references}

\end{document}